\documentclass[prl,superscriptaddress,twocolumn]{revtex4}

\usepackage{amsfonts,amssymb,amsmath}
\usepackage[]{graphics,graphicx,epsfig}
\usepackage{amsthm}

\def\identity{\leavevmode\hbox{\small1\kern-3.8pt\normalsize1}}

\newtheorem{theorem}{Theorem}

\newtheorem{mydef}{Definition}

\renewcommand{\epsilon}{\varepsilon}

\begin{document}

\title{Necessary and sufficient condition for state-independent contextual measurement scenarios}
\author{Ravishankar \surname{Ramanathan}}
\email{ravishankar.r.10@gmail.com}
\affiliation{National Quantum Information Center of Gdansk,  81-824 Sopot, Poland}
\affiliation{University of Gdansk, 80-952 Gdansk, Poland}
\author{Pawel \surname{Horodecki}}
\affiliation{National Quantum Information Center of Gdansk, 81-824 Sopot, Poland}
\affiliation{Faculty of Applied Physics and Mathematics, Technical University of Gdansk, 80-233 Gdansk, Poland}

\begin{abstract}

The problem of identifying measurement scenarios capable of revealing state-independent contextuality in a given Hilbert space dimension is considered. We begin by showing that for any given dimension $d$ and any measurement scenario consisting of projective measurements, (i) the measure of contextuality of a quantum state is entirely determined by its spectrum, so that pure and maximally mixed states represent the two extremes of contextual behavior, and that (ii) state-independent contextuality is equivalent to the contextuality of the maximally mixed state up to a global unitary transformation. We then derive a necessary and sufficient condition for a measurement scenario represented by an orthogonality graph to reveal state-independent contextuality. This condition is given in terms of the fractional chromatic number of the graph $\chi_f(G)$ and is shown to identify all state-independent contextual measurement scenarios including those that go beyond the original Kochen-Specker paradigm \cite{Yu-Oh}.

\end{abstract}
\maketitle

{\it Introduction.} One of the most striking features of the quantum world is contextuality, the notion that outcomes cannot be assigned to measurements independent of the particular contexts in which the measurements are realized. The incompatibility of such assignments manifests itself in the famous Kochen-Specker theorem \cite{KS} which states that for every quantum system belonging to a Hilbert space of dimension greater than two, irrespective of its actual state, a finite set of measurements exists whose results cannot be assigned in a context-independent manner. This phenomenon is known as state-independent contextuality, there are also state-dependent tests \cite{KCBS} in which only a subset of quantum states in a certain Hilbert space display contextual behavior while the results of measurements on other states can be explained by deterministic non-contextual assignments (or their probabilistic mixtures). 

A natural question is: when does a measurement scenario consisting of a set of projective measurements reveal the contextuality of some quantum state? In particular, when does a measurement scenario (up to a global unitary transformation of the projective measurements) reveal the contextuality of all quantum states belonging to a Hilbert space of particular dimension? Some attention has been devoted to these questions \cite{Cabello4, Cabello3}, mainly to the problem of finding the minimal set of measurements that reveals the contextuality of all states of a given dimension \cite{Cabello1, Peres, Cabello2}. It has been shown that the originally studied Kochen-Specker (KS) sets are not the only class of state-independent contextual measurements, there exist measurement scenarios for which all the original KS constraints are obeyed and yet they yield state-independent contextuality (S-IC) \cite{Yu-Oh}. The formulation of general conditions to identify all S-IC measurement scenarios has thus gained importance, more so with the development of applications of contextuality such as in device-independent security \cite{Horodecki2}, in formulating Bell inequalities that are algebraically violated by quantum correlations \cite{Acin2}, in certifying the dimension of quantum systems \cite{Guhne}, in parity-oblivious multiplexing \cite{Spekkens}, etc.

In the study of contextual measurements, a number of tools and invariants of graph theory have appeared \cite{Fritz}. In this paper, we focus on the most commonly used paradigm in the study of contextuality, that of orthogonality graphs. An orthogonality graph representing a set of projective measurements has each projector assigned to a vertex of the graph and edges connect commuting (orthogonal) measurements. The following is a brief summary of the main results in this paper. Firstly, we show that the contextuality of a quantum state as given by the measure of contextuality \cite{Horodecki} in any measurement scenario depends only upon the spectrum of eigenvalues of the state. Consequently, among all states belonging to a given dimension, the maximally mixed state is shown to be the least contextual while the pure states in that dimension are the most contextual for that measurement scenario. This implies that when the contextuality of the maximally mixed state is revealed by a set of projective measurements, the same measurement scenario i.e. a set of measurements represented by the same orthogonality graph upto some global unitary transformation can also reveal the contextuality of \textit{all} quantum states in that dimension. Then the central result of the paper, namely a \textit{necessary and sufficient} condition for a measurement scenario to reveal the contextuality of the maximally mixed state in arbitrary Hilbert space dimension $d$, is shown. This condition is formulated in terms of the graph-theoretic parameter known as the fractional chromatic number $\chi_f(G)$ of the orthogonality graph representing the set of measurements and reveals a common feature of all state-independent contextual orthogonality graphs in relation to their coloring properties.

{\it The class of contextual graphs.} We focus on the orthogonality graph representation of a set of projective measurements $\{\Pi_k\}$, commonly the projectors are of rank one ($\Pi_k  = |v_k \rangle \langle v_k|$). In the orthogonality graph, every projective measurement is denoted by exactly one vertex in the graph and edges connect two vertices if and only if the corresponding projectors are orthogonal. A classical (non-contextual) assignment of outcomes $\{0,1\}$ is then an assignment of $1$'s to the vertices belonging to an independent set of the graph (an independent set is a set of vertices with no edges connecting any pair of vertices) and $0$'s to the other vertices. This assignment guarantees that in each orthonormal basis at most one projector is assigned the value $1$. Orthonormal bases in the graph are represented by maximal cliques which are complete subgraphs of maximal size in which all pairs of vertices are connected by edges. A set of projectors is said to realize the graph $G$ if it strictly obeys the orthogonality constraints of $G$. Given a realization of $G$, it is natural to allow a global unitary transformation on all the projectors since such a transformation preserves the orthogonality relations between projectors and hence preserves $G$. We investigate the contextual properties of orthogonality graphs by considering an optimization over all projectors that realize the particular graph. Note that other representations of contextual measurements by graphs can be translated to the orthogonality graph representation. For instance, one can reformulate in the language of orthogonality graphs one alternative representation found in the literature, namely that of hypergraphs \cite{Peres, Horodecki} where hyperedges consist of commuting observables. One way to do this is to find the common eigenstates of the commuting set of observables in each hyperedge, and to construct the orthogonality graph for the projectors corresponding to these eigenstates.

For an orthogonality graph $G$ with vertex set $\mathcal{V} = \{v_1, \dots v_n\}$, let $\mathcal{I}(G)$ denote the set of all independent sets $I^{(j)}(G)$ of $G$. One may construct corresponding (incidence) vectors $\vec{I}^{(j)}(G) \in \mathbb{R}^n$, where $j \in \{1, \dots, |\mathcal{I}(G)|\}$ with components $\vec{I}^{(j)}_k(G) = 1$ if $v_k \in I^{(j)}(G)$ and $\vec{I}^{(j)}_k(G) = 0$ otherwise. The convex hull of all such vectors $\vec{I}^{(j)}(G)$ then corresponds to the classical non-contextual polytope for the graph $G$ (also known as the stable-set polytope STAB(G)) \cite{Winter}. The set of probabilities realizable in general probabilistic theories consistent with the no-disturbance principle \cite{Monogamy} (in which all the probabilities are well-defined and independent of the context) reside in the polytope known as the clique-constrained stable set polytope QSTAB(G). The notion of clique-constraint refers to the fact that the sum of probabilities for projectors in every orthonormal basis (clique) must not exceed unity. The set of quantum probabilities is a convex set that is in general not a polytope but is contained within QSTAB(G). A quantum state $\rho$ is represented in the polytope by a vector $\vec{x}^{\{\rho\}}(G) \in \mathbb{R}^n$ with components $\vec{x}^{\{\rho\}}_k(G) = tr( \Pi_k \rho)$, corresponding to the probabilities for the optimal projectors $\Pi_k$ obeying $G$ and acting on the state $\rho$. 

Let us now formally state the main ideas corresponding to state-dependent and state-independent contextuality using the notion of orthogonality graphs. For simplicity, we will refer to the states in dimension $d$ i.e. $\rho \in M_{d \times d} (\mathbb{C}^d)$ as $\rho \in \mathbb{C}^d$.

\begin{mydef}
An orthogonality graph $G$ is said to be contextual for dimension $d$ if $\exists \rho \in \mathbb{C}^d$ such that $\vec{x}^{\{\rho\}}(G) \notin STAB(G)$. An orthogonality graph that is not contextual for any dimension is called a non-contextual graph.
\end{mydef}

\begin{mydef}
An orthogonality graph $G$ is said to be state-independent contextual (S-IC) for dimension $d$ if $\forall \rho \in \mathbb{C}^d$, $\vec{x}^{\{\rho\}}(G) \notin STAB(G)$. 
\end{mydef}

In other words, for a state-dependent contextual graph, the points $\vec{x}^{\{\rho\}}(G)$ corresponding to \textit{some} quantum states $\rho$ lie outside STAB(G), while for state-independent graphs the points corresponding to \textit{all} quantum states of given dimension have this property. Intuitively, one might expect that the points corresponding to the pure states are farther from the stable-set polytope than those corresponding to the more mixed states. The following section proves that this is indeed the case by comparing the values of a measure of contextuality \cite{Horodecki} for pure and mixed states and showing that the measure of contextuality follows the majorization order of quantum states.

{\it Measure of contextuality and spectra of states.} Given an orthogonality graph $G$ realizable by a set of projectors $\{\Pi_j\}$ with $\Pi_j \in \mathbb{C}^d$, a quantitative measure \cite{Horodecki} enables us to study the relative contextuality of quantum states $\rho \in \mathbb{C}^d$. Here, we show that the measure of contextuality for any given measurement scenario only depends on the spectrum of eigenvalues of the state. As a consequence, for any $G$ we show that the maximally mixed state $\identity_d$ is the least contextual while the pure states display maximum contextuality among all states of given dimension. 

Allowing a maximization over all sets of projectors $\{\Pi_j\}$ realizing $G$, for given state $\rho$ the distance of $\vec{x}^{\{\rho\}}(G)$ from STAB(G) defines a measure of contextuality of $\rho$ with respect to the graph $G$. This measure denoted $M^{(G)}(\rho)$ is defined as 
\begin{eqnarray}
M^{(G)}(\rho) &:=& \max_{\{\Pi_k\}} M^{(G)}_{\{\Pi_k\}}(\rho)  \\
&:=&\max_{\{\Pi_k\}} \min_{p(\lambda)} \sum_{c \in \mathcal{C}(G)} \frac{1}{|\mathcal{C}(G)|} D(q^{(\rho,\{\Pi_k\})}(\lambda_c) || p(\lambda_c)). \nonumber
\label{measure}
\end{eqnarray}
Here, the maximization is over all sets of projectors $\{ \Pi_k\}$ realizing $G$, $\mathcal{C}(G)$ denotes the set of maximal cliques in the graph $G$ 
and $q^{(\rho,\{\Pi_k\})}(\lambda_c) =\{tr(\rho \Pi_k^{(c)})\}$ is the vector of probabilities corresponding to the clique (context) $c$. The minimization is performed over the classical non-contextual joint probability distributions $p(\lambda)$ for the graph $G$ with $p(\lambda_c)$ being the corresponding marginal distribution for context $c$. The relative entropy distance is $D(q(\lambda_c) || p(\lambda_c)) = \sum_{i} q(\lambda_c)_{i} \log \frac{q(\lambda_c)_i}{p(\lambda_c)_i}$ where summation extends over all projectors belonging to the context $c$ and any (possibly different rank) projectors needed to realize a complete measurement ($\sum_i q(\lambda_c)_i = \sum_i p(\lambda_c)_i = 1$). The measure thus captures the notion of distance from the stable-set polytope. Let $\vec{\kappa}(\rho)$ denote the spectrum of eigenvalues of $\rho$, we then have the following. 

\begin{theorem}
$M^{(G)}(\rho)$ is fully determined by $\vec{\kappa}(\rho)$. Moreover, $\forall \rho_1, \rho_2 \in \mathbb{C}^d$ and $\forall G$, $\vec{\kappa}(\rho_1) \prec \vec{\kappa}(\rho_2)$ $\Rightarrow$ $M^{(G)}(\rho_1) \leq M^{(G)}(\rho_2)$. In particular, $M^{(G)}(\identity_d) \leq M^{(G)}(\rho)$, $\forall \rho \in \mathbb{C}^d$. 
\end{theorem}

\begin{proof}
Consider the maximization procedure in Eq.(\ref{measure}). Two states $\eta_1, \eta_2$ sharing the same spectrum $\vec{\kappa}(\eta)$ are unitarily equivalent, $\eta_1 = U \eta_2 U^{\dagger}$ for some unitary $U$. Therefore, given a set of projectors $\{\Pi_k\}(\eta_1)$ that result from the maximization in the computation of the measure for $\eta_1$, one may construct a corresponding set of projectors $\{U^{\dagger} \Pi_k U \}$ that gives the same value of the measure for $\eta_2$ (and vice versa). 

Given $\rho_1, \rho_2 \in \mathcal{H}^d$, if their vector of eigenvalues obey the majorization relation $\vec{\kappa}(\rho_1) \prec \vec{\kappa}(\rho_2)$, it then follows \cite{Horn}  that there exist $d$-dimensional permutation matrices $P_j$ and a probability distribution $\{p_j\}$ such that $$\vec{\kappa}(\rho_1) = \sum_j p_j P_j \vec{\kappa}(\rho_2).$$ Given an orthogonal graph $G$, let us denote the set of corresponding projectors resulting from the maximization in the measure $M^{(G)}(\rho_1)$ by $\{\Pi_k\}(\rho_1)$ as above. With $\rho_2$ written in diagonal form, let $\rho_{2}^{j} = P_j \rho_2 P_j^{\dagger}$ denote the diagonal matrices obtained from the permutation of the eigenbasis of $\rho_2$. For each $\rho_2^{j}$, let the joint probability distribution resulting from the minimization in Eq.(\ref{measure}) with the chosen set of projectors $\{\Pi_k\}(\rho_1)$ be denoted by $p^{\rho_2^{j}}(\lambda)$. We then have,
\begin{widetext}
\begin{eqnarray}
M^{(G)}(\rho_2) \geq \sum_{j} p_{j} M^{(G)}_{\{\Pi_k\}(\rho_1)}(\rho_2^{j}) &=& \sum_{c \in \mathcal{C}(G)} \frac{1}{|\mathcal{C}(G)|} \sum_j p_j \left(D(q^{(\rho_2^{j},\{\Pi_k\}(\rho_1))}(\lambda_c) \Vert p^{\rho_2^{j}}(\lambda_c))\right) \nonumber \\
&\geq& \sum_{c \in \mathcal{C}(G)} \frac{1}{|\mathcal{C}(G)|} \left(D(\sum_j p_j q^{(\rho_2^{j},\{\Pi_k\}(\rho_1))}(\lambda_c) \Vert \sum_j p_j p^{\rho_2^{j}}(\lambda_c))\right) \nonumber \\
&\geq& \sum_{c \in \mathcal{C}(G)} \frac{1}{|\mathcal{C}(G)|}  \left(D(q^{(\rho_1,\{\Pi_k\}(\rho_1))|}(\lambda_c) \Vert \sum_j p_j p^{\rho_2^{j}}(\lambda_c))\right) \nonumber \\
&\geq& M^{(G)}(\rho_1). \nonumber 
\end{eqnarray}
\end{widetext}
In the above, the first inequality is a result of the fact that $\{\Pi_k\}(\rho_1)$ may not yield the maximum in the computation of the measure for $\rho_2$, the second inequality is a result of the convexity property of the relative entropy, the third follows from linearity, and the fourth by the definition of the measure for state $\rho_1$. The monotonicity of the measure of contextuality with majorization order thus follows. 

Since the spectrum of the maximally mixed state $\identity_d$ is majorized by that of every other state $\rho \in \mathcal{H}^d$, it follows that $\identity_d$ is the least contextual among all states belonging to a given Hilbert space and analogously, pure states are the most contextual.
\end{proof}

The theorem shows that the ordering of states according to the measure of contextuality follows the majorization order for any measurement scenario $G$. A similar analysis as in the proof above shows that not only the measure of contextuality, but also the violation of non-contextuality inequalities (defining the facets of STAB(G)) by quantum states follows the majorization order. In other words, $\vec{\kappa}(\rho_1) \prec \vec{\kappa}(\rho_2)$ implies that the violation of any non-contextuality inequality by $\rho_1$ does not exceed its violation by $\rho_2$ when we allow an optimization over projectors realizing $G$. 

Knowing that the maximally mixed state is the least contextual among all states in the given Hilbert space for any $G$, the identification of state-independent contextual measurement scenarios can be reduced to the question: for given $G$, does $\vec{x}^{\{\identity_d\}}(G)$ belong outside the stable-set polytope? As we shall see in the following, there is one graph-theoretic quantity (the fractional chromatic number) that precisely identifies when this happens. 

{\it State-independent contextual graphs.} 
We now proceed to formulate a necessary and sufficient condition for an orthogonality graph $G$ realizable in dimension $d$ to be state-independent contextual in that dimension in terms of its fractional chromatic number $\chi_f(G)$ \cite{Scheinerman}, defined below using the notion of fractional colorings. We show that this condition is strictly stronger than an analogous necessary condition proposed in \cite{Cabello3} and present an explicit counter-example showing that the latter condition is not sufficient.
\begin{mydef}
A fractional coloring of a graph $G$ \cite{Scheinerman} is a non-negative real-valued function $f$ on $\mathcal{I}(G)$ such that for any vertex $v$ of $G$, $\sum_{s \in \mathcal{I}(G,v)} f(s) \geq 1$, where $\mathcal{I}(G,v)$ denotes the set of independent sets of $G$ that contain vertex $v$. The weight of a fractional coloring is $w = \sum_{j=1}^{\mathcal{I}(G)} f(I^{(j)}(G))$, and $\chi_f(G) = \min w$. Equivalently, $\chi_f(G) = \min \frac{a}{b}$ s.t. there is a coloring of $G$ using $\texttt{a}$ colors that assigns $\texttt{b}$ colors to each vertex with adjacent vertices getting disjoint sets of colors.
\end{mydef}

For any graph it is known that $\omega(G) \leq \chi_f(G) \leq \chi(G)$, where the clique number $\omega(G)$ is the size of the maximal clique in the graph. The first inequality comes from the fact that in a clique of size $\omega(G)$, assigning $b$ colors to each vertex requires at least $a = \omega(G) b$ colors; the second inequality can be seen from the fact that $\chi(G)$ is the particular restriction $b=1$ in the definition of $\chi_f(G)$ above. In fact, there is a practical way to compute $\chi_f(G)$ by means of a linear program. One assigns weights $w_{j}$ to each independent set $I^{j}(G) \in \mathcal{I}(G)$. Then $\chi_f(G) = \min\{\sum_{j} w_{j}: \sum_{j} w_{j} \vec{I}^{j}(G) \geq \vec{\bf{1}}\}$

\begin{theorem}
An orthogonal graph $G$ realizable by a set of rank-$r$ projective measurements $\{\Pi_j\}$ with $\Pi_j \in \mathbb{C}^d$ is state-independent contextual for dimension $d$ if and only if $\chi_f(G) > \frac{d}{r}$. 
\end{theorem}

The proof of the above theorem is provided in the supplemental material. The importance of the above necessary and sufficient condition stems from the fact that it enables one to identify sets of projective measurements that reveal the contextuality of all states in a given dimension, upto some global unitary transformations. In particular, it is powerful enough to identify not only the Kochen-Specker measurement configurations but also the recently found surprising alternative proofs of state-independent contextuality where all the Kochen-Specker constraints are obeyed \cite{Yu-Oh}.


{\it Discussion.}
The identification of measurement configurations that reveal state-independent contextuality has gained importance owing to the recent development of applications of contextuality. One related necessary condition has been derived in \cite{Cabello3} where it was shown that a set of rank-1 projective measurements reveals the contextuality of the maximally mixed state in dimension $d$ if $\chi(G) > d$. Now, we will show that this condition is not sufficient to identify state-independent contextual graphs by means of an explicit counter-example. 

Consider the state-independent orthogonality graph in dimension $3$, $G_{YO}$ considered in \cite{Yu-Oh}. The graph $G_{YO}$ is a $13$-vertex graph with $\chi(G_{YO}) = 4$ and $\xi(G_{YO}) = 3$. Here $\xi(G)$ denote the smallest dimension in which the graph $G$ may be realized by rank-1 projectors with distinct vertices assigned distinct projectors and two vertices connected by an edge if the corresponding projectors are orthogonal. The fractional chromatic number of the graph $G_{YO}$ can be computed by means of a linear program to be $\chi_f(G_{YO}) = 3 \frac{2}{11}$ so that it is indeed state-independent contextual in dimension $3$. The counter-example to the sufficiency of the chromatic number condition is provided by the join of two copies of the graph $G_{YO}$ (the join of two graphs is their graph union with additional edges connecting all the vertices of one graph with those of the other). The join ($J$) of two $G_{YO}$ graphs has $\chi(J(G_{YO},G_{YO})) = 8$ and  $\xi(J(G_{YO},G_{YO})) = 6$. Moreover, its fractional chromatic number can be computed to be $\chi_f(J(G_{YO},G_{YO})) = 6\frac{4}{11}$ so that in dimension $7$ this graph is not by itself state-independent contextual as $\vec{x}^{\{\identity_7\}}(J(G_{YO}, G_{YO})) \in STAB(J(G_{YO},G_{YO}))$. An explicit decomposition of the point $\vec{x}^{\{\identity_7\}}(J(G_{YO}, G_{YO}))$ in terms of a convex combination of deterministic points is provided in the Supplementary Material. 

As the counter-example shows, $\chi_f(G) > d$ does not necessarily imply $\chi(G) > d$ and the latter is only a necessary condition for $G$ to be state-independent contextual. 
Moreover, it is known that the gap between $\chi_f(G)$ and $\chi(G)$ can be arbitrarily large \cite{Scheinerman} so that $\chi_f(G)$ is by itself a better indicator of the S-IC property of $G$. In general, $\chi_f(G)$ is computed as a lower bound to $\chi(G)$ by means of the fractional relaxation to the integer linear program required to compute the latter quantity. However, the fact that it can be computed by a linear program does not imply that $\chi_f(G)$ can be computed in polynomial time, since the number of independent sets may be exponential in the number of vertices. An interesting possibility is that $\chi(G) > d$ may constitute a sufficient condition for state-independent contextuality of a measurement scenario $G_c$, formed from $G$ by adding projectors so that each clique constitutes a complete orthonormal basis, this is left as an open question.

{\it Conclusions.} We have investigated the requirements for a measurement scenario consisting of a set of projective measurements to reveal state-independent contextuality. 
In terms of the rigorous measure of contextuality, for any measurement scenario represented by an orthogonality graph $G$, the maximally mixed state was shown to be the least and pure states the most contextual among all states of a given Hilbert space dimension. A necessary and sufficient condition for identifying state-independent contextual measurement scenarios was formulated in terms of the fractional chromatic number $\chi_f(G)$. A similarly formulated condition in terms of $\chi(G)$ fails to guarantee sufficiency of the S-IC property for $G$. $\chi_f(G)$ being in general easier to compute than $\chi(G)$, the formulated condition enables the identification of state-independent contextual measurements which have been utilized as resources in various scenarios \cite{Horodecki2, Acin2, Guhne, Spekkens}.

{\em Acknowledgements}. We thank Prof. Ryszard Horodecki for a critical reading of the manuscript. This work is supported by the ERC grant QOLAPS.  

{\bf Supplementary Material}

Here, we present proof of Theorem 2 formulated in the main text along with an explanation of the counter-example to the condition based on the chromatic number \cite{Cabello3}.

\textit{\textbf{Theorem 2.} An orthogonality graph $G$ realizable by a set of rank-$r$ projective measurements $\{\Pi_j\}$ with $\Pi_j \in \mathbb{C}^d$ is state-independent contextual for dimension $d$ if and only if $\chi_f(G) > \frac{d}{r}$. }
\begin{proof}
For given set of rank-$r$ projective measurements $\{\Pi_j\}$ in $d$-dimensional space, let us construct the orthogonality graph $G$ with a total of $|\{\Pi_j\}| = n$ vertices. The vector of expectation values for the maximally mixed state $\vec{x}^{\{\identity_{d}\}}(G)$ is then the uniform vector with components  $\vec{x}^{\{\identity_{d}\}}_k(G) = \frac{r}{d}$ for $1 \leq k \leq n$. The fractional chromatic number is by definition $\chi_f(G) =\min{\frac{a}{b}}$ where out of a total of $'a'$ colors, $'b'$ colors are assigned to each vertex such that vertices connected by an edge receive disjoint sets of colors. Let $\{a^*(G), b^*(G)\}$ denote the values of $a$ and $b$ that achieve the minimum in this definition. For each of the colors $1 \leq c \leq a^*(G)$, let $\vec{I}^{(c)}(G)$ denote the incidence vector of the independent set of vertices colored with color $c$ and $\vec{I}^{(0)}(G)$ denote the vector with all components $0$. We may consider three cases.

(i) $\chi_f(G)= \frac{a^*}{b^*} = \frac{d}{r}$. 

In this case, the fractional colorings themselves yield the convex decomposition 
$$\vec{x}^{\{\identity_{d}\}}(G) = \sum_{c=1}^{a^*} \frac{1}{a^*}\vec{I}^{(c)}(G).$$
Since each vertex in the decomposition appears in a total of $b^*$ independent sets, the uniform vector representing the maximally mixed state is exactly reproduced. 

(ii) $\chi_f(G) = \frac{a^*}{b^*} < \frac{d}{r}$.

Let $a' = \frac{b^* d}{r} > a^*$. Then one may construct the following decomposition into non-contextual assignments
$$\vec{x}^{\{\identity_{d}\}}(G) = \sum_{c=1}^{a^*} \frac{1}{a'}\vec{I}^{(c)}(G) + (1-\frac{a^*}{a'}) \vec{I}^{(0)}(G)$$
with the final vector appearing to ensure the probabilities sum to unity. Therefore, once again we have that $\vec{x}^{\{\identity_{d}\}}(G) \in STAB(G)$.

(iii) $\chi_f(G) > \frac{d}{r}$. 

To prove that in this case, we definitely have $\vec{x}^{\{\identity_{d}\}}(G) \notin STAB(G)$, we may use the formulation of $\chi_f(G)$ in terms of a linear program. Let us assign weights $w_j$ to each independent set $I^{(j)}(G) \in \mathcal{I}(G)$. Then $\chi_f(G) = min \{\sum_{j=1}^{|\mathcal{I}(G)|} w_{j}: \sum_{j=1}^{|\mathcal{I}(G)|} w_{j} \vec{I}^{(j)}(G) \geq \vec{\bf{1}}\}$ where $\vec{\bf{1}}$ denotes the uniform vector of $1$'s on each vertex. 
The following Lemma \cite{Scheinerman} then enables us to change the inequality in the definition to equality. 

\textit{Lemma:} If a graph $G$ has a fractional coloring with total weight $w = \sum_{j=1}^{|\mathcal{I}(G)|} w_j$ such that $\sum_{j=1}^{|\mathcal{I}(G)|} w_{j} \vec{I}^{(j)}(G) \geq \vec{\bf{1}}\}$, then one may construct a fractional coloring with weight $w' \leq w$ such that  $\sum_{j=1}^{|\mathcal{I}(G)|} w'_{j} \vec{I}^{(j)}(G) = \vec{\bf{1}}\}$.

Assume $\vec{x}^{\{\identity_{d}\}}(G) \in STAB(G)$, and let the following be a convex decomposition for $\vec{x}^{\{\identity_{d}\}}(G)$ into non-contextual assignments. 
$$ \vec{x}^{\{\identity_{d}\}}(G) = \sum_{j=1}^{\mathcal{I}(G)} p_j \vec{I}^{(j)}(G).$$ Then since the components $\vec{x}^{\{\identity_{d}\}}_k(G) = \frac{r}{d}$ for $1 \leq k \leq n$, we have $\sum_{j=1}^{\mathcal{I}(G)} \frac{p_j d}{r}\vec{I}^{(j)} = \vec{\bf{1}}$ so that $\chi_f(G) \leq \sum_{j=1}^{\mathcal{I}(G)} \frac{p_j d}{r} = \frac{d}{r}$ which is a contradiction. 
Therefore, $\chi_f(G) > \frac{d}{r}$ is a necessary and sufficient condition for $\vec{x}^{\{\identity_{d}\}}(G) \notin STAB(G)$ (in other words, for $G$ to be state-independent contextual).
\end{proof}


{\it II. An explicit decomposition of $\vec{x}^{\{\identity_7\}}(J(G_{YO}, G_{YO}))$.}
Consider two copies of the state-independent graph $G_{YO}$ considered in \cite{Yu-Oh} and reproduced in Fig. ({\ref{yuoh}}). The join of the two copies is defined as the graph $J(G_{YO}, G_{YO})$ obtained by the union of the two graphs with additional edges connecting every vertex of the first copy with every vertex of the second. 

\begin{figure}[t]
\vspace{-1cm}
\begin{center}
\includegraphics[scale=0.3]{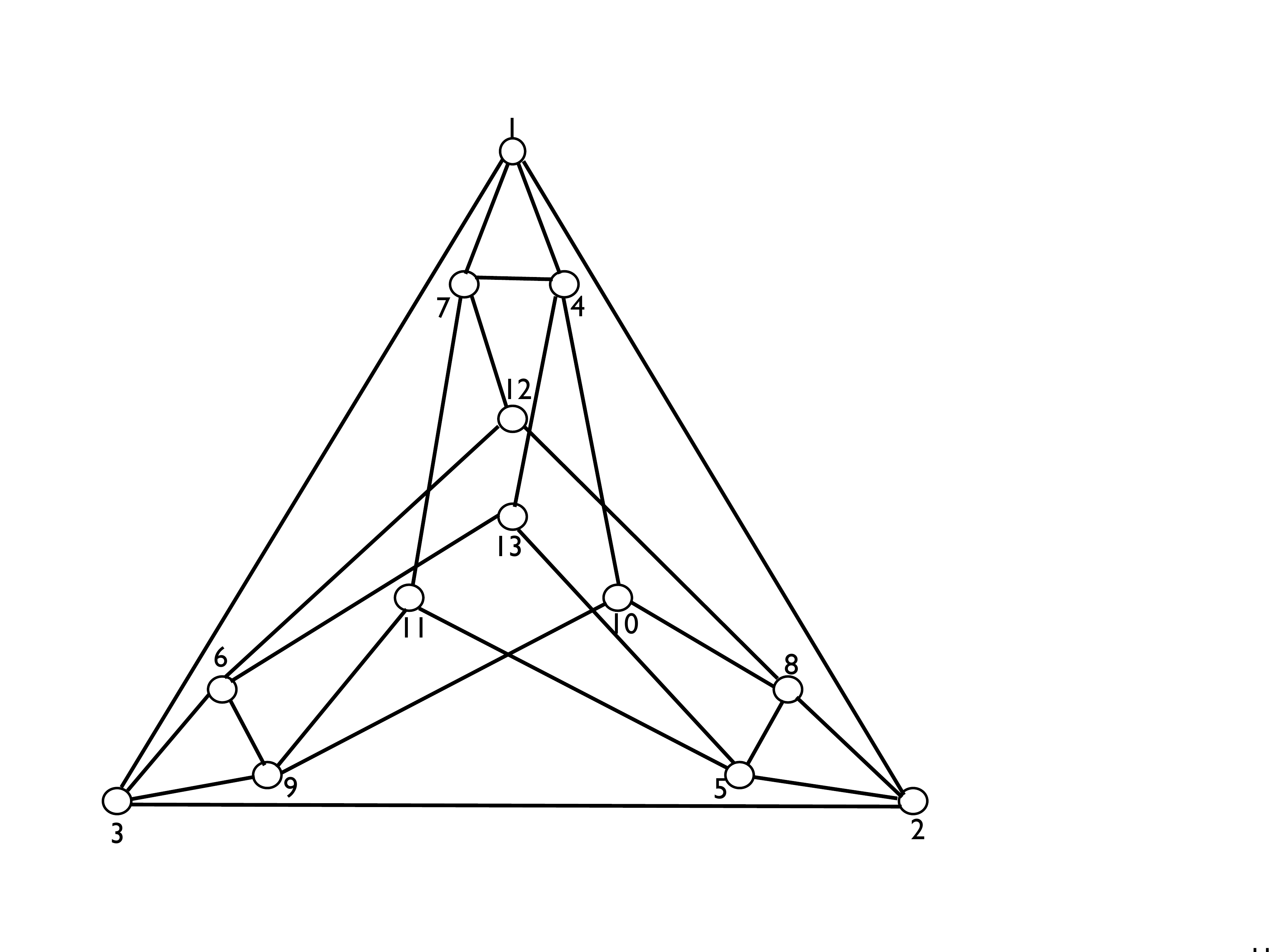}
\end{center}
\vspace{-1cm}
\caption{The $13$ vertex graph $G_{YO}$. The join of two copies of this graph provides a counter-example to $\chi(G) > d$ being a sufficient condition for a graph to be state-independent contextual by itself.}
\label{yuoh}
\end{figure}

Let us label the vertices $\{1, \dots, 13\}$ in one copy as in Fig. (\ref{yuoh}) and $\{1',\dots,13'\}$ in the second copy. For notational convenience, let us use $\vec{I}(\{v_i\})$ to denote the incidence vector of the independent set $\{v_i\}$. $J(G_{YO}, G_{YO})$ has $\omega(J(G_{YO}, G_{YO})) = 6$ and $\chi(J(G_{YO}, G_{YO})) = 8$ so one may expect it to state-independent contextual (by itself without completion) for dimension $7$. However, the following explicit convex decomposition of the vector $\vec{x}^{\{\identity_7\}}(J(G_{YO}, G_{YO}))$ proves otherwise.
\begin{widetext}
\begin{eqnarray}
\vec{x}^{\{\identity_7\}}(J(G_{YO}, G_{YO})) &=&\frac{1}{77} \big[ 5 \vec{I}(\{1,5,6,10\})  +  4 \vec{I}(\{1,5,9,12\}) + \vec{I}(\{1,6,8,11\}) + \vec{I}(\{1,10,11,12,13\}) + 2 \vec{I}(\{2,4,6,11\}) \nonumber \\ &+& 2 \vec{I}(\{2,4,9,12\}) + 3 \vec{I}(\{2,6,7,10\})  + 3 \vec{I}(\{2,7,9,13\}) + \vec{I}(\{2,10,11,12,13\}) + 5 \vec{I}(\{3,4,8,11\}) \nonumber \\ &+& 5 \vec{I}(\{3,7,8,13\}) + \vec{I}(\{3,10,11,12,13\}) + 2 \vec{I}(\{4,5,9,12\}) \big]  + \frac{1}{77} \left(v_i \leftrightarrow v'_i\right) + \frac{1}{11}\vec{I}^0 \nonumber
\end{eqnarray}
\end{widetext}
Here, $\left(v_i \leftrightarrow v'_i\right) $ denotes the same combination of independent sets considered on the second copy of $G_{YO}$ and $\vec{I}^0$ denotes the vector with all vertices assigned value $0$. From the above, note that $\chi_f(J(G_{YO}, G_{YO})) = 6\frac{4}{11}$. Note that this does not rule out the possibility that a completed version of $J(G_{YO}, G_{YO})$ in dimension $7$ is a state-independent graph.

\end{document}